\documentclass[11pt]{article}
\usepackage[numbers]{natbib}
\usepackage{fullpage}
\usepackage{float}
\usepackage{epsfig}
\usepackage{amsmath}
\usepackage{amsthm}
\usepackage{amssymb}
\usepackage{amstext}
\usepackage{xspace}
\usepackage{latexsym}
\usepackage{verbatim}
\usepackage{multirow}
\usepackage{ifthen}
\usepackage{url}
\usepackage{aliascnt}
\usepackage{enumerate}
\usepackage[pdfpagelabels,pdfpagemode=None]{hyperref}

\usepackage[ruled]{algorithm2e}

\SetAlFnt{\small}
\SetAlCapFnt{\small}
\SetAlCapNameFnt{\small}
\SetAlCapHSkip{0pt}
\IncMargin{-\parindent}

\newtheorem{theorem}{Theorem}[section]

\newaliascnt{lemma}{theorem}
\newtheorem{lemma}[lemma]{Lemma}
\aliascntresetthe{lemma}

\newaliascnt{claim}{theorem}
\aliascntresetthe{claim}

\newaliascnt{corollary}{theorem}
\newtheorem{corollary}[corollary]{Corollary}%
\aliascntresetthe{corollary}

\newaliascnt{proposition}{theorem}
\aliascntresetthe{proposition}

\theoremstyle{definition}
\newtheorem{definition}{Definition}

\newtheorem{remark}{Remark}

\newenvironment{proofof}[1]{\vspace{0.1in} {\sc Proof of #1.}}{\hfill\qed}
\newenvironment{proofsk}{\vspace{0.1in} {\sc Proof sketch.}}{\hfill\qed}

\renewcommand{\comment}[1]{}

\newcommand{\AutoAdjust}[3]{\mathchoice{ \left #1 #2  \right #3}{#1 #2 #3}{#1 #2 #3}{#1 #2 #3} }
\newcommand{\Xcomment}[1]{{}}

\newcommand{\InBrackets}[1]{\AutoAdjust{[}{#1}{]}}
\newcommand{\Ex}[2][]{\operatorname{\mathbf E}_{#1}\InBrackets{#2}}

\def\expect{\Ex}

\newcommand{\argmax}{\operatorname{argmax}}
\newcommand{\argmin}{\operatorname{argmin}}

\newcommand{\feasset}{\mathcal I}

\newcommand{\val}{v}

\newcommand{\vali}[1][i]{\val_{#1}}
\newcommand{\alloc}{x}

\newcommand{\payment}{p}

\newcommand{\dist}{F}

\DeclareMathOperator{\REV}{Rev}

\newcommand{\reserve}{r}
\newcommand{\monres}{\reserve^{\mathrm{mon}}}
\newcommand{\winset}{W}

\newcommand{\vcgprice}{\payment^{\mathrm{VCG}}}

\newcommand{\highlazy}{h^{\mathrm{L}}_{-i}}
\newcommand{\higheager}{h^{\mathrm{E}}_{-i}}
\newcommand{\paymenteager}{\payment^{\mathrm{E}}}

\newcommand{\conddist}{\dist^{\mathrm{cond}}}

\newcommand{\paymentvcgl}{\payment^{\mathrm{VCGL}}}
\newcommand{\paymentvcge}{\payment^{\mathrm{VCGE}}}

\newcommand{\losesetprime}{W'}
\newcommand{\basecommon}{U}
\newcommand{\matroid}{\mathcal M}
\newcommand{\highrev}{H}
\newcommand{\lowrev}{L}
\newcommand{\larev}{\REV}

\DeclareMathOperator{\rev}{R}

\newcommand{\ratio}{\rho}
\newcommand{\ratiocond}{\rho^{\mathrm{cond}}}

\newcommand{\br}{\mathbf r}
\newcommand{\bs}{\mathbf s}
\newcommand{\bS}{\mathbf S}
\newcommand{\bsminusi}{\bs_{-i}}
\newcommand{\bSminusi}{\bS_{-i}}

\newcommand{\bF}{\mathbf F}

\newcommand{\set}{Z}
\newcommand{\tilR}{\tilde{R}}

\newcommand{\GVCGL}{GVCG-L$^*$}

\title{Approximate Revenue Maximization in Interdependent Value Settings
\thanks{Part of this work was done while the first and third authors were
visiting Microsoft Research. The first author is supported in part by NSF grants
CCF-1101429 and CCF-1320854, and the third author by NSF grant CCF-1016509.}} 

\author{
Shuchi Chawla \\
University of Wisconsin - Madison \\
Computer Sciences Department \\
\tt{shuchi@cs.wisc.edu}
\and Hu Fu \\
Microsoft Research \\
New England \\
\tt{hufu@microsoft.com}
\and Anna R. Karlin \\
University of Washington \\
Department of Computer Science and Engineering \\
\tt{karlin@cs.washington.edu}
}

\begin{document}

\begin{titlepage}
\maketitle
\begin{abstract}

  We study revenue maximization in settings where agents' values are
  interdependent: each agent receives a signal drawn from a
  correlated distribution and agents' values are functions of all of
  the signals. We introduce a variant of the generalized VCG auction
  with reserve prices and random admission, and show that this auction
  gives a constant approximation to the optimal expected revenue in
  matroid environments. Our results do not require any assumptions on
  the signal distributions, however, they require the value functions
  to satisfy a standard single-crossing property and a concavity-type
  condition.

\end{abstract}
\thispagestyle{empty}
\end{titlepage}

\section{Introduction}
Revenue maximization, a.k.a. {\em optimal auction} design, is a fundamental goal in mechanism design. 
The most common approach to the problem is Bayesian wherein it is assumed that players' values are
drawn from prior distributions known to the
designer. The seminal work of 
\citet{M81} on this topic assumes the agents' values are drawn from independent distributions. In this paper, we consider settings in which this independence assumption is unreasonable.  If the value of the item derives from its resale value, such as a house, the values of the agents are likely to be correlated.  Yet another example is fashion, whether haute couture or fancy cars, where demand is highly correlated.
Likewise, when the value of the item derives from the potential for making money from it,
the values are correlated, or perhaps even common. 
The classical example is an auction for the right to drill for oil in a certain
location~\cite{Wilson69}. Note though that even in settings such as this where the value is common, bidders might have different information about what that value actually is. For example, the value of an oil lease depends on how much oil there actually is, and the different bidders may have access to different assessments about this. Consequently, a bidder might change her own estimate of the value of the oil lease given access to the information another bidder has.

The following model due to \citet{MW82} has become standard for  auction
design in these correlated and common value settings, generally known as {\em
interdependent settings}: \citep[See also][]{Krishna10,Milgrom04}
\begin{itemize}
\item Each agent has a real-valued, private {\em signal}, say $s_i$
  for agent $i$. The set of signals $\bs = (s_1, s_2, \ldots, s_n)$ is
  assumed to be drawn from a known, correlated distribution.

The signals are meant to capture the information available to the
agents about the item. For example, in the setting of a painting, the
signal could be information about the provenance of the painting or
private information from experts. In the setting of oil drilling
rights, the signals could be information that engineers have about the
site based on geologic surveys, etc.
\item The {\em value of the item} to agent $i$ is a function
  $v_i(\bs)$ of the information or signals of {\em all} agents.  
  
  For example, in the {\em mineral rights model} \citep{Wilson69}, each
  signal $S_i= V + X_i$ might be a noisy version of the true value $V$
  of the oil field (say $X_i$ is $N(0,1)$), and $v_i(\bs)= E(V | s_1,
  \ldots, s_n)$. 

Another important case is when $v_i(\bs) = s_i +
  \beta \sum_{j\ne i}s_j$, for some $\beta \le 1$. This type of
  valuation function captures settings where an agent's value depends
  on his own personal appreciation for the item ($s_i$) and on resale
  value (as reflected by the total appreciation of the other
  agents) \citep{M81,Klemp98}.

\end{itemize}

Interdependent settings have received attention in the economics community for
quite some time \citep[see][]{Krishna10}.\footnote{Within the theoretical
computer science community, interdependent settings have received far less
attention and most of that only in the last few years.  \citep[See,
e.g.][]{R01,PP11, DFK11,BKL12,AABG13,Syrgkanis13,RT13,Li13}.} Perhaps the most
notable results here are due to~\citet{CM85,CM88} \citep[see also][]{MMR89,
MR92, Rahman14}. These papers consider single item auctions and show that if the
prior distribution from which $\bs$ is drawn satisfies a type of full rank
condition, and the valuation functions are known and satisfy a single-crossing
condition  (see Section \ref{sec:prelim}), then the auctioneer can construct an
auction that essentially extracts full surplus (that is, has expected revenue
equal to $\Ex{\max_i v_i (S_1, \ldots, S_n)}$.

From a practical perspective, the Cr\'emer-McLean auction is problematic in two
significant ways: First, the auction is not ex post individually rational.  That
is, after the fact, agents may regret having participated, and indeed may end up
being charged a very large amount to participate.  Second, the dependence 
on the signal distribution is complex in form and  can be numerically unstable. Moreover, agents need to know the prior distribution from which the signals are drawn.
For these reasons, this result has been criticized as being impractical
\citep[e.g., in][]{Milgrom04}.

These shortcomings have driven much recent progress in the
understanding of optimal auctions.  Addressing the first issue, i.e.,
strengthening the requirement for individual rationality so that
bidders are guaranteed non-negative utilities under any circumstance,
a series of recent papers \citep{Segal03,CE06,CE07, Vohra11, Li13, Lopomo00}
studies \emph{ex post equilibria},\footnote{Bidding truthfully is an
  {\em ex post equilibrium} if an agent does not regret having bid
  truthfully when the mechanism terminates, given that other agents
  bid truthfully. In other words, bidding truthfully is a Nash
  equilibrium at the end.  See Section~\ref{sec:prelim}.} culminating
in \citet{RT13}'s characterization of the expected revenue as the
``conditional virtual surplus''.  This characterization was used to
derive optimal auctions for symmetric matroid settings with various
assumptions on the distribution and the value functions (see
Section~\ref{sec:prelim} for details).  In particular, for the general interdependent setting, the value
distribution is required to satisfy a \emph{monotone hazard rate (MHR)} condition, a
property that holds for certain unimodal distributions with
light tails.

\vspace{0.1in}
\parbox{0.95\textwidth}{ In this work, we design auctions for a broad
  class of interdependent value settings that obtain approximately
  optimal revenue in ex post equilibrium in the absence of any
  distributional assumptions.} \par
\vspace{0.1in}

Our auctions are variants of the VCG auction with reserve
prices.  The VCG auction with reserve prices is known to be approximately optimal in independent private value settings under certain assumptions.  Myerson's characterization of the optimal auction implies
that VCG with monopoly reserve prices is optimal in the i.i.d. values
setting under a matroid feasibility constraint. In non-i.i.d. settings, the
optimal auction can be complex and may be impractical. \citet{HR09}
initiated the study of VCG auctions with reserve prices, as a model
for simple and yet approximately optimal auctions.

For interdependent settings that satisfy a single-crossing condition
(a condition also assumed by \citet{RT13} and virtually all literature
on these settings), \citet{Li13} showed that a generalized notion of
the VCG auction with appropriate reserve prices gives an
$e$-approximation to the optimal revenue for matroid settings with
generalized MHR distributions.  Unfortunately, for more general
distributions, such reserve-based VCG auctions perform poorly (see
example in Section~\ref{sec:intro-gen} below).



We introduce a new variant of VCG auctions: a random admission phase
followed by a VCG auction with reserve prices. We show that this
randomization preserves the incentive compatibility and the
approximately optimal performance of VCG auctions even when all
restrictions on the value distribution are lifted, although we will
need a certain condition on the value function.  As a special case of
our main result, we have:

\begin{theorem}
\label{thm:lookahead-general-interdependent-approx}
The randomized \GVCGL\ auction described in Section~\ref{sec:general}
is ex post incentive compatible, individually rational, and obtains a
constant fraction of the optimal revenue under a matroid feasibility
constraint, assuming that the valuation functions are 
additively separable functions of signals, and satisfy the single-crossing condition.
\end{theorem}

Theorem~\ref{thm:lookahead-general-interdependent-approx} in fact
holds for a larger class of valuation functions that
includes, e.g., concave functions of additively separable functions of
signals (see Assumption (A.3) in Section~\ref{sec:prelim}). Although
Theorem~\ref{thm:lookahead-general-interdependent-approx} can be seen
as a simple vs.\@ optimal type of result, we emphasize that, for
non-MHR value distributions, no optimal auctions or approximately
optimal ones were previously known, even for selling a single item.

We now describe our results in more detail.

\subsection{Correlated Private Values}

We first consider the special case where the values ($v_i := v_i(\bs)
= s_i$) are privately known and drawn from a correlated distribution
$F$. The computational complexity of designing the revenue maximizing
mechanism in this setting at ex post Nash equilibrium\footnote{For
  correlated values, ex post equilibrium and dominant-strategy
  equilibrium coincide.  See Section~\ref{sec:prelim}.}  under the
assumption of ex post individual rationality has been largely
resolved. \citet{PP11} proved that the deterministic revenue maximization is
computationally hard even with 3 bidders. In some settings positive
results were obtained: \citet{DFK11} showed how to derive truthful-in-expectation mechanisms via linear programming when the joint
distribution of values is given explicitly; \citet{RT13} give optimal
auctions for matroid settings under the assumptions of affiliation and
a strong version of regularity.

\paragraph{Lookahead auctions for a single item.}

\citet{R01} took the approximation approach,
and proposed the following intuitve \emph{lookahead} auction:
choose the highest bidder as the tentative winner and then run an optimal
auction for the highest bidder conditioning on all other bidders' valuations and
the fact that this tentative winner's valuation is above all others'.  It is not
hard to see that this auction is dominant strategy incentive-compabible for
correlated bidders.  Moreover, with a simple and elegant proof, Ronen showed
that this auction gives at least half of the optimal revenue of any single-item
auction. \citet{DFK11} later extended these to a more general class of
lookahead auctions with better approximation ratios.

Our first result generalizes the lookahead auction to the matroid
setting for correlated bidders in the most natural way, and shows that
it is still a $2$-approximation.

\begin{theorem}
\label{thm:lookahead-matroid-approx-1}
The lookahead auction in Section~\ref{sec:corr} gives at least half of
the optimal revenue in any matroid setting for bidders with
arbitrarily correlated valuations.
\end{theorem}

It is interesting to note that the lookahead auction (in both the
single-item case and in our generalization for the matroid case) is a
special case of \emph{VCG-L auctions}: A VCG-L auction \citep{DRY10, Li13} is an auction that first selects as tentative winners the
bidders who would be winners in a VCG auction, and then offers a
take-it-or-leave-it price to each of these tentative winners, where
the price is the higher of the VCG price and a preset reserve
price. In the lookahead auction, the reserve price is determined
optimally based on the conditional value distribution of the tentative
winner, conditioned on others' revealed values and on being the
tentative winner; We call this the {\em conditional monopoly
  reserve}. In fact, the lookahead auction or VCG-L with conditional
monopoly reserves is the optimal among all VCG-L auctions.

\subsection{General Interdependent Settings}
\label{sec:intro-gen}

General interdependent settings are far less understood than private value
settings.  Indeed, even the more basic problem of maximizing social
welfare is not straightforward: It is no longer possible to design
dominant-strategy auctions, since an agent's value depends on all the
signals, and so if, say, agent~$i$ reports rubbish, agent $j$ might
win at a price above his value if he reports truthfully.  Hence, the
standard approach is to try maximizing welfare at ex post equilibria.
Unfortunately, maximizing social welfare in an ex post equilibrium is
also provably impossible unless the valuation functions $v_i(\bs)$
satisfy a so-called {\em single-crossing condition}.  This condition
means, in essence, that the influence of the highest bidder's signal
on her own value is at least as high as its influence on the other
high bidders' values.\footnote{ This implies that, given signals
  $\bs_{-i}$, if agent $i$ has the highest value when $s_i = t^*$,
  then agent $i$ continues to have the highest value for $s_i >
  t^*$.}  When the single-crossing condition holds, there is a
generalization of VCG auctions that maximizes efficiency at ex post
equilibrium for matroid
settings~\cite{CM85,CM88, Krishna10, Ausubel99, Dasgupta00, Li13}:
the mechanism asks agents to reveal their signals, computes values
based on the reported signals, and then runs a VCG-type auction on these values.

Returning to the goal of approximate revenue maximization, a generalization of
VCG-L with conditional monopoly reserves was proposed by~\citet{Li13}.  This auction,
which we will call \emph{Generalized VCG-L with conditional monopoly
  reserves} or \GVCGL\ for short, is the following: 
\begin{enumerate}
\item Ask agents to report their signals; compute agents' values.
\item Choose the social welfare-maximizing set $W$, and determine the prices
they would be charged in the generalized VCG auction.
\item Offer each agent~$i$ in $W$ a take-it-or-leave-it price which is the maximum of the price from
step (2), and the optimal price conditioned on $i$ being in $W$ and conditioned on the signals of all the other agents.
\end{enumerate} 
This auction was shown by~\citet{Li13} to obtain an expected revenue
which approximates the optimal social welfare (and hence optimal
revenue as well) for matroid settings under the assumption that the
conditional distributions satisfy the monotone hazard rate (MHR)
condition.  \citet{Li13}'s analysis relies on a property peculiar to
MHR distributions, namely that the optimal revenue from a single agent
is close to the agent's expected value.

Unfortunately, \GVCGL\ does not give any constant approximation to the
optimal revenue if the conditional distributions are not MHR, even for
single item settings. Consider, for instance, a setting with two
agents and one item where the first agent's value  is $v_1(s_1, s_2) = s_1$  and the second agent's value is $v_2(s_1, s_2) = s_1 - \epsilon$. In this case, \GVCGL\ always serves the
first agent, obtaining the same expected revenue as in a single agent auction
with this agent alone, because the signal~$s_2$ contains no information and
$v_2$ is completely determined by~$s_1$ (incentive compability prevents
the use of $v_2$ in extracting revenue from the first agent).  
The optimal auction, on the other hand always serves the second
agent obtaining her entire value (because $s_1$ completely determines
$v_2$). The gap between the revenues of these two auctions can be
arbitrarily large if the distribution of $s_1$ is non-MHR. We note, in
particular, that a revenue-maximizing auction must at times sell to an
agent whose value is not the largest. Accordingly, we add a stage of
random admission at the start of the auction, ensuring that the
auction precludes the highest-value agent from winning with constant probability,
while still using her signal to help with setting the prices for the
other agents.  This enables us to overcome the shortcoming of the
\GVCGL\ auction and to recover revenue from all cases with constant
probability.

\subsection{Results in Independent Private Value Settings}

\citet{R01}'s analysis of the lookahead auction uses first principles
without appealing to \citeauthor{M81}'s virtual values. Taking the
same approach allows us to develop new insight into the design of
simple approximately-optimal auctions even in {\em independent private
  value} settings. In \autoref{app0} we rederive several known
simple vs.\@ optimal results \citep[e.g.,][]{HR09, ADMW13} as quick consequences
of the lookahead auction.  We single out one step
of this in \autoref{app1}, which may be of independent interest.  
This concerns the VCG-E or \emph{VCG with eager reserves} auction, which
differs from the VCG-L auction only in that it removes all bidders
bidding below their reserve prices {\em before} picking the tentative
winners.  We show that, when the reserve prices are at least the
monopoly reserves, VCG-E is always preferable than VCG-L, in terms of
both social welfare and revenue (Theorem~\ref{thm:eager-vs-lazy}).


\subsection{Other Related Work}
\label{sec:related}

We briefly give further details on some of the related work.  

\paragraph{The standard auctions.}
The revenue in Bayes-Nash equilibrium of the standard auctions (English, first
price and second price) has been analyzed assuming a strong form of positive
correlation known as ``affiliation'' between the signals \citep{MW82}. One of the most interesting results here is the ``linkage principle'' which shows that in the symmetric, affiliated setting, the auctioneer benefits by enabling the maximum amount of information to be revealed. A consequence is that open auctions such as the English auction generally lead to higher prices than closed auctions such as sealed-bid first price auctions.

\paragraph{Characterization of revenue in ex post equilibrium.}

As mentioned earlier, a series of papers, including \citep{Segal03,CE06,CE07,
Vohra11, Li13, RT13, Lopomo00} \citep[see][for a literature review]{RT13}, develop characterization results for ex post equilibrium that show 
that the expected revenue in ex post equilibrium is equal to the expected ``conditional virtual surplus''.

\paragraph{Optimal auctions in ex post equilibrium.}
The characterization results just mentioned are used to derive
Myerson-like optimal auctions in some settings. The most broadly
applicable results here are due to \citet{RT13} who use their
characterization to derive optimal auctions (a) for correlated values
in matroid settings, with assumptions of regularity and affiliation,
and (b) for the general interdependent case when bidders are
symmetric, have affiliated signals, and satisfy a number of other
conditions~\citep{Lopomo00}.  They also derive the first
approximately-optimal prior-independent results under the Lopomo
assumptions for matroid feasibility constraints. An auction is {\em
  prior independent} if it does not require knowing the distributions
of agents' values \citep{DRY10}.


Several of these papers \citep[e.g.,][]{Lopomo00, Neeman03, CE07, RT13} focus specifically on the English auction and derive among other things that the English auction with a suitable reserve price is optimal among all ex post IC and ex post IR auctions  for symmetric settings with correlated values that satisfy regularity and affiliation.

\paragraph{Bayes-Nash equilibrium versus ex post equilibrium.}
We refer the reader to \citet{RT13} for an excellent discussion of the tradeoffs between BIC mechanism design and ex post IC mechanism design.

\section{Preliminaries}
\label{sec:prelim}

\paragraph{Single Dimensional Environments.}

In a single dimensional auction environment, an auctioneer is offering a service (or goods) to $n$ bidders, but with certain constraints on which subsets of bidders can be simultaneously served.  Formally, if we use $[n]$ to denote the set of bidders, then there is a set $\feasset \subseteq 2^{[n]}$, such that the auctioneer can serve a subset $S$ of bidders simultaneously if and only if $S$ is in~$\feasset$.  For example, in a single item auction, $\feasset$ consists of all singleton subsets of~$[n]$ and the empty set.  We say the environment is \emph{downward closed} if $T \in \feasset$ implies $S \in \feasset$, $\forall S \subseteq T$.

A feasibility constraint $\feasset$ is called a {\em matroid constraint} if it is downward closed and for all $A,B\in\feasset$ with $|A|>|B|$, there exists $e\in A$ such that $B\cup\{e\}\in\feasset$. Matroid settings encompass important types of markets, e.g.\@ digital goods (when all subsets are feasible, i.e., $\feasset = 2^{[n]}$), $k$ unit auctions (when $\feasset$ contains all subsets of size at most~$k$, i.e., is the $k$-uniform matroid), and unit-demand bipartite matching markets (when $\feasset$ is a transversal matroid).  


\paragraph{Signals, valuations and distributions.}

Each bidder~$i$ has a private one-dimensional signal~$s_i$ representing the
information available to him.  The signal $s_i$ is a realization of a random variable $S_i$. We will assume that  the signals $\bS=(S_1, \ldots, S_n)$ are drawn from a  known \emph{correlated} joint distribution denoted by $\dist$. We denote the vector  $(s_1, \ldots, s_{i - 1}, s'_i, s_{i + 1}, \ldots, s_n)$ by
$(s'_i, \bs_{-i})$.

An agent's valuation $v_i(\bs)$ is a function of all the signals. (In the special case of correlated or independent private values, $v_i$ is equal to the signal $s_i$ and in that case, we do not refer to signals, but rather to the valuations $v_i$ themselves.) These valuation functions are assumed to be common knowledge, and satisfy the following conditions:

\begin{itemize}
\item[(A.1)] For all $i$, $v_i(\bs)$ is increasing in each coordinate and strictly increasing in $s_i$.
\item[(A.2)] The valuations satisfy the {\em single-crossing condition}:
For each $i\ne j$, $\bs_{-i}$, $s_i$, and $s'_i$, with $s'_i>s_i$,
$v_i(s_i, \bs_{-i})\ge v_j(s_i, \bs_{-i})$ implies $v_i(s'_i, \bs_{-i})> v_j(s'_i, \bs_{-i})$. In other words, as soon as agent $i$'s value crosses agent $j$'s value, then increasing $i$'s signal continues to keep $v_i$ larger than $v_j$.
\item[(A.3)] The responsiveness of an agent's value to someone else's signal decreases as the agent's own signal increases. Formally,  
for all $i\ne j$ and for all $\bs$,
  $\frac{\partial v_i(\bs)}{\partial s_j}$ is a non-increasing
  function of $s_i$. 
\end{itemize}

Assumptions (A.1) and (A.2) are standard in interdependent value settings. Assumption (A.3) is new. One broad class of value functions for which this assumption holds is {\em additively separable} valuations, i.e. $v_i(\bs) = \sum_j g_{ij}(s_j)$ where each $g_{ij}(\cdot)$ is a non-decreasing functions. More generally, if values are concave functions of additively separable functions of signals, then assumption (A.3) holds.


\paragraph{Auctions and Incentive Compatibility.}

An \emph{auction} takes as input a signal from each bidder and produces as output a pair of vector functions $(\alloc, \payment)$ defined on tuples of signals.  At the signals
$(s_1, \ldots, s_n)$, the \emph{allocation} function, $\alloc_i(s_1, \ldots, s_n)$,  gives the probability with which
bidder~$i$ receives the service, and $\payment_i(s_1, \ldots, s_n)$ , the \emph{payment} function, denotes the expected payment
bidder~$i$ makes to the auctioneer. (The randomization here is in the mechanism.)

Bidder~$i$'s \emph{utility} from participating in the auction $(\alloc(\cdot), \payment(\cdot))$, when the true signals are $\bs$, he reports $r_i$, and the other bidders report $\br_{-i}$, is $v_i (\bs)\alloc_i(\br) - \payment_i(\br)$. 

Since bidders' signals are private, they may misreport their signals to gain advantage.  In order to incentivize
truthtelling, we study auctions that are \emph{incentive compatible}. 
Three different equilibrium notions, presented in order of decreasing strength, will come up in this paper:  
\begin{itemize}
\item An auction is called \emph{dominant strategy
incentive compatible} (DSIC) if, for all  $i$, true signals $(s_1, \ldots, s_n)$, reported signals $(r_1, \ldots, r_n)$ and $s_i'$:
\begin{align*}
v_i (\bs)\alloc_i(s_i, \br_{-i}) - \payment_i(s_i, \br_{-i}) \geq v_i(\bs) \alloc_i(s'_i, \br_{-i}) - \payment_i(s'_i,\br_{-i} ).
\end{align*}

\item An auction is called \emph{ex post
incentive compatible} (ex~post IC, for short)\footnote{All ex~post IC auctions
that we consider will, in fact, satisfy the ex~post IC condition even after any internal randomization in the auction.}
 if, for all $i$, true signals
$(s_1, \ldots, s_n)$ and $s_i'$
\begin{align*}
v_i (\bs)\alloc_i(\bs) - \payment_i(\bs) \geq v_i(\bs) \alloc_i(s_i', \bs_{-i}) - \payment_i(s_i', \bs_{-i}).
\end{align*}

\item An auction is called \emph{Bayesian incentive compatible} (BIC) if, for all $i, s_i$ and $s'_i$,
$$\Ex{v_i (s_i,\bSminusi)\alloc_i(s_i, \bSminusi) - \payment_i(s_i, \bSminusi)} \geq \Ex{v_i (s_i,\bSminusi)\alloc_i(s'_i, \bSminusi) - \payment_i(s'_i, \bSminusi)}.$$

\end{itemize}
In this paper we will focus only on incentive compatible auctions, and so we will use the terms \emph{bids} and
\emph{signals}, or, in the correlated case, \emph{valuations}, interchangeably.

We will also focus exclusively on auctions that satisfy  \emph{ex post
individual rationality}, that is each agent's final utility is  always nonnegative.

 \paragraph{Social Welfare and Revenue.}

The \emph{social welfare} for serving a set $T$ of bidders with signals $\bs$ is $\sum_{i \in T} v_i(\bs)$.  The expected revenue
of an auction is $\Ex{\sum_i \payment_i(\bS)}$.



\paragraph{VCG and Generalized VCG auctions with lazy reserve prices.}

In a VCG auction \citep{Vic61, Clarke71, Groves73},  the winning set of bidders
is the feasible set with the largest  sum of valuations. Each of them pays the
critical valuation below which he would drop from the winning set. We will be
interested in the generalization of VCG that applies in interdependent matroid
settings that satisfy the single-crossing condition \citep{CM85, CM88,
Krishna10, Ausubel99, Dasgupta00, Li13}.

 Formally, in generalized VCG,
the set $\winset$ of winners is 
$$\winset:=\argmax_{T \in \feasset} \sum_{i \in T} v_i(\bs).$$ To define
payments, let$$s^*_i =   \inf_{s_i} \{i \in \argmax_{T \in \feasset} \sum_{j \in T} v_j(s_i, \bsminusi)\}.$$
The  \emph{VCG payment} or \emph{VCG threshold} for bidder~$i$ is then
 $$\vcgprice_i :=v_i (s^*_i, \bsminusi).$$
This
payment $\vcgprice_i$ is  defined solely by the signals of other
bidders, the form of the valuation functions, and the feasibility system~$\feasset$, and therefore is well defined for all bidders, including those who lose
the auction.  For each bidder, the VCG threshold is the value above which he will win the auction, and below which he
loses.  

By construction, the VCG auction maximizes social welfare. Moreover, for correlated settings, VCG is dominant strategy incentive compatible, and for interdependent settings, generalized VCG is ex post incentive compatible. (See Lemma~\ref{lem:GVCGexpost}.)

The main object of study in this paper is a generalized VCG-type auction with individual reserve prices.  

\begin{definition}[Generalized VCG Auctions with Lazy Reserve Prices (GVCG-L)] 
In the generalized VCG auction with \emph{lazy reserve prices} $(\reserve_1, \ldots, \reserve_n)$, agents are asked to report their signals. Given these signals, the mechanism first chooses a tentative subset $\winset$ of winners as in
the GVCG auction, i.e., $\winset =\argmax_{T \in \feasset} \sum_{i \in T} v_i(\bs)$.  Then each bidder $i \in \winset$ is
presented a take-it-or-leave-it price set to be the higher of $\reserve_i$ and $\vcgprice_i$.
\end{definition}

\begin{lemma}~\citep{Li13}
\label{lem:GVCGexpost}
For any interdependent setting with matroid feasibility, the Generalized VCG auction with lazy reserve prices is ex post incentive compatible as long as for each $i$, the reserve price $\reserve_i$ is specified independently of $s_i$. For correlated, downward-closed settings, it is dominant-strategy incentive compatible.
\end{lemma}

For completeness, we provide a brief proof sketch for this lemma.

\begin{proofsk}
  It suffices to prove that GVCG without reserve prices is ex post IC, in other
words, for all agents $i$, $i$ is in the winning set $W$ if and only if his
value exceeds his threshold $\vcgprice_i(\bs_{-i})$. Recall that for a matroid
constraint, the winning set $W$ can be found by ordering agents by decreasing
value and greedily selecting a maximal feasible set in this order. Also, in a
matroid, the sizes of all maximal feasible subsets of a given set are equal. It
follows then that whether or not an agent $i$ belongs to the set $W$ depends
only on the set of agents preceding $i$ in the ordering by decreasing value, and
not on the relative ordering of those agents. Furthermore, if an agent in the
winning set $W$ unilaterally raises his signal, the single-crossing property
implies that his rank in the ordering improves. Putting these two observations
together we may conclude that the agent belongs to $W$ as long as his signal
exceeds his threshold $s^*_i$.
\end{proofsk}

\section{Correlated private values}
\label{sec:corr}
For correlated private value single-item auctions, \citet{R01}
proposed the lookahead auction which he showed $2$-approximates the
optimal revenue.  In fact, the lookahead auction is simply VCG-L with
conditional monopoly reserves, that is, where the reserve price for
the highest bidder is set optimally based on the conditional distribution of
the agent's value given others' values and the fact that the agent
has the maximum value. We now give a natural extension of this
result to matroid settings: we show that VCG-L with conditional
monopoly reserves continues to give a 2-approximation to expected
revenue in these settings.



\begin{theorem}{\em(Restatement of Theorem~\ref{thm:lookahead-matroid-approx-1}.)}
\label{thm:lookahead-matroid-approx}
The VCG-L auction with conditional monopoly reserves obtains at least
half of the optimal revenue under a matroid feasibility constraint
when agents have correlated private values.
\end{theorem}

To prove \autoref{thm:lookahead-matroid-approx}, we will need the next well-known theorem on matroids.  \citep[See, e.g.][for a proof.]{S03}

\begin{theorem}
\label{thm:matroid-mapping}
Let $B_1$ and $B_2$ be any two independent sets of a matroid~$\matroid$ such that $|B_1| = |B_2|$.  There exists a bijective
mapping $g: B_1 \setminus B_2 \to B_2 \setminus B_1$ such that $\forall e \in B_1 \setminus B_2$, $B_2 \setminus \{e\} \cup
\{g(e)\}$ is independent in~$\matroid$. 
\end{theorem}

\begin{proof}[Proof of \autoref{thm:lookahead-matroid-approx}]

  We use VCG-L$^*$ to denote VCG-L with conditional monopoly
  reserves. Denote by $\larev$ the expected revenue of the VCG-L$^*$
  auction.  The revenue of any optimal auction can be split into two
  parts, $\highrev$ and~$\lowrev$: $\highrev$ is the expected revenue
  from~$\winset$, the set of tentative winners, and $\lowrev$ the
  revenue from the rest of the agents.  Note that $\winset$ here is
  random, determined by the realization of bidders' valuations, but
  $\highrev$ and $\lowrev$ are expected values and not random.  It
  suffices to show that $\larev$ is no less than both $\highrev$
  and~$\lowrev$.

$\larev$ is clearly at least $\highrev$, since the VCG-L$^*$ auction runs the
optimal auction for each agent in~$\winset$, using all information available at that stage.

Let $\losesetprime \subseteq [n] \setminus \winset$ be an independent subset that maximizes the social welfare among bidders not
in~$\winset$.  The expectation of $\sum_{j \in \losesetprime} v_j$ is an upper bound for
$\lowrev$, since the auction cannot charge more than the agents' valuations.  Therefore it suffices to show $\larev \geq \Ex{\sum_{j \in \losesetprime} v_j}$.
Since $|\losesetprime| \leq |\winset|$, we can find a subset $\basecommon \subseteq \winset$ such that $\basecommon \cup \losesetprime$ is independent and $|\basecommon \cup
\losesetprime| = |\winset|$.  By \autoref{thm:matroid-mapping} there exists a bijective mapping $g: \winset \setminus
\basecommon \to \losesetprime$ such that for any bidder~$i$ in $\winset \setminus \basecommon$, $\winset \setminus \{i\} \cup
\{g(i)\}$ is independent.  Therefore, the VCG threshold $\vcgprice_i$ for each $i \in \winset \setminus \basecommon$ is at least
$v_{g(i)}$.  In the second stage of the VCG-L$^*$ auction, if the auctioneer simply
sets VCG payment for each agent in~$\winset$,  a revenue of $\sum_{j \in \losesetprime} v_j$
would have been secured.  By the optimality of the revenue
from~$\winset$ in the lookahead auction, we have $\larev \geq \Ex{\sum_{j \in \losesetprime} v_j} \geq \lowrev$.
\end{proof}

As we point out in \autoref{app0}, it is not hard to see that, when the private valuations are
independently drawn from regular distributions, the reserve prices optimally set
for an agent in~$\winset$ will be either the VCG payment or the
\emph{monopoly reserve}, the optimal price one would set when selling a
single item to this single bidder.  In this case,
\autoref{thm:lookahead-matroid-approx} directly implies that the VCG-L auction
with monopoly reserves $2$-approximates the optimal auction, a result first
shown by \citet{DRY10}.  In \autoref{app0} and \autoref{app1} we also present
similarly quick derivations of several other simple vs.\@ optimal results from
the literature.

The proof for \autoref{thm:lookahead-matroid-approx}, generalized from
\citeauthor{R01}'s proof for the lookahead auction, is direct and simple,
without appealing to special properties of the distribution, nor any
characterizations of the expected revenue (e.g.\@ as virtual surplus).  We use
this as a convenient building block in the design and analysis of our mechanisms
for interdependent value settings, which nonetheless call for new ideas to
overcome the plain lookahead auctions' limitations.

\section{General interdependent values}
\label{sec:general}
We now consider the general interdependent setting where agents'
values depend on their own as well as on others' signals. We consider
a natural generalization of the mechanism presented in the previous
section for the correlated private values setting: generalized VCG-L
with conditional monopoly reserves, or \GVCGL. We saw an example in
the introduction where the \GVCGL\ mechanism does not always obtain a
good approximation to the optimal revenue, even with just a single
item for sale, simple value functions, and ``nice'' signal
distributions.  We get around the shortcomings in \GVCGL\ by running
the mechanism on a random subset of the agents. Importantly, we ensure
that the agent with the highest value is left out of this subset with
constant probability, so that the mechanism can extract revenue from
agents with lower values as well. We will now show that this
modification is sufficient to obtain a constant factor approximation
to optimal revenue.

Before we describe our mechanisms formally, we introduce some notation
and prove a bound on the expected revenue of any ex post IC
mechanism. Fix a vector of reported signals $\bs$ (which since we consider only incentive-compatible mechanisms are the true signals). We will use $v(A)$
to denote the total value of a subset $A$ of agents: $\sum_{i\in A}
v_i(\bs)$. Let $W$ be the maximum value feasible set (e.g., in the
case of a matroid constraint, an optimal base in the matroid). Let
$W'$ be the maximum value feasible set that is disjoint from $W$.

\begin{lemma}
\label{lem:OPT-bound}
  Let $W$ and $W'$ be defined as above. Then the expected revenue of
  any ex post IC mechanism is bounded by
$$\expect{\sum_{i\in W} R_i + v(W')}$$
where the expectation is over the randomness in signals; $R_i$ is the
expected revenue obtained by offering to serve agent $i$ at the
monopoly reserve price for the conditional value distribution of $i$,
conditioned on $\bs_{-i}$ and on $i$ being in $W$.
\end{lemma}
\begin{proof}
  The revenue that any ex post IR mechanism can extract from agents
  not in $W$ is at most $v(W')$. Therefore, we focus on the revenue
  that a mechanism can extract from the set $W$. This revenue is
  bounded from above by the sum of revenues of $n$ ex post IC
  mechanisms, the $i$th one of which serves only agent $i$, and serves
  this agent only when this agent belongs to the set $W$. Any
  single-agent ex post IC mechanism is a (random) posted price
  mechanism with a posted price that depends on the reported signals
  of the other agents; therefore, the revenue obtained by such a mechanism
  is at most $R_i$.
\end{proof}

\subsection{Single item setting}
We now describe our mechanism for the single item setting. Consider
the following mechanism that we call randomized \GVCGL:

\vspace{0.1in}
\framebox[0.95\textwidth][c]{
\parbox{0.9\textwidth}
{
\begin{enumerate}
\item Ask each agent to report his signal.
\item Include each agent in a set $\set$ with probability 2/3.
\item Run \GVCGL\ on $\set$; when determining the reserve price for agent~$i$,
condition on signals of all agents except~$i$ (including those not in $\set$).  
\end{enumerate}
}
}

\begin{theorem}
  Consider a single item setting with interdependent values, where the
  value functions satisfy assumptions (A.1)--(A.3). The randomized \GVCGL\
  is ex post IC and achieves a 4.5-approximation to the optimal ex
  post IC mechanism.
\end{theorem}

\begin{proof}
  Let $\bs = (s_1, s_2, \ldots, s_n)$ denote the agents'
  signals. Without loss of generality, suppose that $v_1 (\bs) \ge
  v_2(\bs) \ge v_j(\bs)$ for all $j > 2$.  Define
$$s_1^* =\argmin_s \{v_1(s, \bs_{-1}) \ge v_j(s, \bs_{-1})\quad
\forall j > 1\},$$ In other words, $s_1^*$ is the smallest signal at
which agent 1's value is the highest, given $\bs_{-1}$. Let agent $i$
be the one who determines the threshold $s_1^*$, that is, $v_1(s_1^*,
\bs_{-1})= v_i(s_1^*, \bs_{-1})$. Note that $i$ may or may not be the
agent with the second highest value at $\bs$.

Lemma~\ref{lem:OPT-bound} implies that



$$OPT \le	R_1 + v_2 (s_1, s_{-1})$$
where $R_1$ is the expected revenue from agent 1 conditioned on $s_1
\ge s_1^*$, and $\bs_{-1}$. 

On the other hand, the expected revenue of randomized \GVCGL\ is at
least
$$Pr( 1 \in \set, i\in \set) \cdot R_1 + Pr (2 \in \set, 1 \not \in \set) \cdot R_2 = \frac{2}{9} (2R_1 + R_2)$$
where $R_2$ is the expected revenue from agent 2 conditioned on the fact that $v_2(\bs)$ is
highest among agents in $\set$, and conditioned on
$\bs_{-2}$. Therefore, randomized \GVCGL\ recovers within a small
factor the first component of the optimal revenue.

We will now bound the second term in $OPT$, namely $v_2 (s_1, \bs_{-1})$,
in terms of $R_1$ and $R_2$.
Observe that fixing $\bs_{-1}$, agent 1's value at his threshold
$s_1^*$ is at least as large as any agent's value. Furthermore, $R_1$
extracts at least this threshold value. That is,
$$v_2(s_1^*, \bs_{-1}) \le v_1 (s_1^*, \bs _{-1}) \le R_1.$$
Next, using assumption (A.3) we get that
$$v_2 (s_1, \bs_{-1}) - v_2 (s_1^*, \bs_{-1})\le
v_2(s_1,0,\bs_{-12})-v_2(s_1^*,0,\bs_{-12}) \le
v_2(s_1,0,\bs_{-12})\le R_2$$ 
where the third inequality follows from
noting that agent 2's value conditioned on others' signals is at least
$v_2(s_1,0,\bs_{-12})$. Therefore, we have $v_2 (s_1, s_{-1})\le
R_1+R_2$, and $OPT\le 2R_1+R_2$.
\end{proof}

\subsection{Matroid setting}

Next we consider settings with a matroid feasibility constraint.\footnote{We
note that our approach also applies to knapsack constraints with the
modification that in \GVCGL\ instead of picking the maximum value
feasible set in the first step we greedily pick a maximal feasible set.}
We will run the \GVCGL\ mechanism on a random subset $\set$ of
agents. \GVCGL\ will select the maximum value feasible subset of this
set, call it $T$, and optimally price the agents in that set
conditioned on others' signals. Define $T'=T\setminus W$. Our goal in
selecting $\set$ is two-fold. First, we want to ensure that $T$
contains agents in $W$ with high probability, so that we can recover
(to some approximation) the $R_i$'s for agents $i$ in $W$. Second, we
want to ensure that $T'=T\setminus W$ contains value comparable to
that in $W'$. We will then be able to ``charge'' the value of every agent in
$T'$ to the revenue we recover from this agent plus the revenue we
recover from some agent in $W$. With these goals in mind, we define
the mechanism as follows.



\vspace{0.1in}
\framebox[0.95\textwidth][c]{
\parbox{0.9\textwidth}
{
\begin{enumerate}
\item Ask each agent to report his signal.
\item With probability 1/2 let $\set$ be the set of all agents; and with probability
1/2 include each agent independently in $\set$ with probability 1/2.
\item Run \GVCGL\ on $\set$; when determining the reserve price for an
agent~$i$, condition on the signals of all agents except~$i$ (including those
not in~$\set$).
\end{enumerate}
}
}
\vspace{0.1in}


The following two lemmas capture the properties that we require from
$T$ and $T'$.
\begin{lemma} 
\label{lem:compare} 
$\expect{v(T')}\ge  \frac{1}{8} \expect{v(W')}$.
\end{lemma}

\begin{lemma} 
\label{lem:map} For any random $T'$ there exists an injective mapping $f_{T'}$ from the elements of $T'$ to the elements of $W$ 
such that for each $j\in T'$ and $i=f_{T'}(j)$, we have,
$$v_j(s_i^*, \bs_{-i}) \le v_i(s_i^*, \bs_{-i}).$$
\end{lemma}

Before proving the lemmas, we give a proof of our main theorem:

\begin{theorem}
  Consider an interdependent values setting with a matroid feasibility
  constraint where the value functions satisfy assumptions
  (A.1)--(A.3). The randomized \GVCGL\ is ex post IC and achieves a
  18-approximation to the optimal ex post IC mechanism.
\end{theorem}

\begin{proof}
  Recall that for an agent $i\in W$, $R_i$ denotes the expected
  revenue we can obtain from this agent by setting an optimal reserve
  price for him conditioned on the other agents' signals and on being
  in the set $W$. With probability $1/2$ our mechanism runs \GVCGL\ on
  all of the agents, and therefore obtains this revenue $R_i$ from
  every agent $i\in W$. With the remaining probability, the mechanism
  obtains some revenue from agents in $W$ and some revenue from agents
  in $T'=T\setminus W$. For agents $j$ not in $W$, we can lower bound
  the revenue that the mechanism can collect from them if they end up
  in $T'$: this is the expected revenue obtained by optimally pricing
  agent $j$ conditioned on $\bs_{-j}$;\footnote{Note that in this lower
    bound we do not condition on $j$ being in set $T'$; the extra
    conditioning can potentially increase the expected revenue, but
    this increase is difficult to estimate because it depends on the
    random set $\set$.} call this quantity $\tilR_j$.

We therefore get that the expected revenue of the mechanism is at
least 
$$\frac{1}{2} \sum_{i\in W} R_i  + \expect{\sum_{j\in T' }\tilR_j }.$$
On the other hand, recall that by Lemma~\ref{lem:OPT-bound} we can bound
the optimal revenue as
$$OPT \le \sum_{i\in W} R_i +v(W').$$
We will now upperbound the second term  in terms of the $R_i$'s and $\tilR_j$'s. 
Fix a random choice of $T'$.
Let $j\in T'$ and $i= f_{T'}(j)$, as defined in Lemma \ref{lem:map}.
Then $$\tilR_{j} \ge v_j(0, \bs_{-j}) \ge v_j(s_i,0, \bs_{-ij}) - v_j(s_i^*,0, \bs_{-ij}) \ge v_{j} (s_i, \bs_{-i}) - v_{j} (s_i^*, \bs_{-i})$$
where the last inequality follows from assumption (A.3).
But by Lemma \ref{lem:map}, for that same $i= f_{T'}(j)$,
$$R_i \ge v_i (s_i^*, \bs _{-i})\ge v_{j}(s_i^*, \bs_{-i}),$$
so
$$\tilR_j + R_i \ge v_j(\bs).$$
Finally, by Lemma \ref{lem:compare},
\begin{eqnarray*}
v(W') & \le & 8 \expect{\sum_{j\in T'} v_j(\bs)}\\
&\le & 8 \expect{\sum_{j\in T'}(\tilR_j + R_{f_{T'}(j)})}.
\end{eqnarray*}
Putting it together:
\begin{eqnarray*}
OPT &\le& \sum_{i\in W} R_i + v(W')\\
&\le & \sum_{i\in W} R_i + 8 \expect{\sum_{j\in T'}(\tilR_j + R_{f_{T'}(j)})}\\
&\le & 9\sum_{i\in W} R_i  +8\expect{\sum_{j\in
    T'}\tilR_j }, 
\end{eqnarray*}
and we get an $18$-approximation.
\end{proof}

It remains to prove the lemmas, for which we will need the following fact whose proof can be found e.g., in
 \citet{S03}.
 
 \begin{lemma}[Strong Basis Exchange]
Let $B$ and $B'$ be two bases of a matroid. Then for all $x\in B\setminus B'$, there is a $y \in B'\setminus B$ such that both $B-x + y$ and $B' - y + x$ are bases.
\end{lemma}

\begin{proofof}{Lemma~\ref{lem:compare}}
In step (2) with probability $1/2$ we construct $\set$ by sampling
each agent. Let us condition on this event.
We will fix a particular order in which to sample elements for $Z$  and in the
process inductively update a particular independent set. It will end up containing all the elements of $W\cap Z$, and each element of $W'$ will be in it, and in $Z$, with probability 1/4.

Assume that $|W'|= |W|$.  (If not, extend $W'$ to a full base using elements of $W$ and adapt the following arguments appropriately.)
Fix an ordering on the elements of $W$, say
$a_1 \ldots, a_r$.  We will also be considering the elements of
$W'$ in a particular order $b_1, \ldots, b_r$ to be determined.

As we sample elements for $Z$, we will inductively update two (ordered) sets of matroid elements $A_i$ and $B_i$ that satisfy the following invariants: 
\begin{itemize}
\item $A_i$ and $B_i$ are both bases.
\item The sampling has been done only for elements $\{a_1, \ldots, a_i\}$ and $\{b_1, \ldots, b_i\}$.
\item The first $i$ elements in $A_i$ and $B_i$ are the same.
\item All elements in $\{a_1, \ldots, a_i\}\cap Z$ are in $A_i$.
\item Each element in $\{b_1, \ldots, b_i\}$ is in $Z\cap A_i$ with probability 1/4.
\end{itemize}
Initially $A_0= W$ and $B_0 = W'$ satisfy the above invariants.
Now suppose that $A_{i-1}$ and $B_{i-1}$ have been constructed and satisfy the above invariants.
Apply the strong basis exchange lemma to find an element  $b_i \in B_{i-1}$ such that $A_{i-1} - a_i + b_i$ and
$B_{i-1} - b_i + a_i$ are both bases. 

Now toss the coin to determine if $a_i \in Z$. If so, $A_i := A_{i-1}$
and $B_i  := B_{i-1} - b_i + a_i$. If $a_i \not\in Z$, then
$A_i := A_{i-1} - a_i + b_i$ and $B_i := B_{i-1}$. Finally toss the coin
to determine if $b_i\in Z$. It is easy to see that all the required invariants are satisfied.

At the end of the process, consider the set $A^*=A_r\cap Z$ (where $r$ is the
rank of the matroid). Note that
all of the agents in $W\cap Z$ belong to $A_r$, so they also belong to
$A^*$. On the other hand, the greedy algorithm for matroids includes
all agents in $W\cap Z$ into $T$. So, $A^*\setminus W$ is a candidate for the
set $T'=T\setminus W$. We now note that $A^*\setminus W$
contains sufficient value: by the invariant above, every element of
$W'$ is in $A_r\cap Z$ with probability $1/4$. Therefore we get:
\[ \expect{v(T')} \ge \expect{v(W'\cap A^*)} = \frac{1}{4} v(W').\]
Removing the conditioning specified earlier yields the lemma.
\end{proofof}

\begin{proofof}{Lemma~\ref{lem:map}}
 Let the rank of the matroid be $r$. Recall that $|T|\le r$ and
  $T'=T\setminus W$.  For now assume that $T'$ is full rank.
 We construct the function $f_{T'}$ by  showing that there is 
 a matching between $W$ and $T'$ such that for each $i$ in $W$, its match $j= M(i)\in T'$ has $v_j(s_i^*, \bs_{-i}) \le v_i(s_i^*, \bs_{-i}).$
 
 To this end, construct a bipartite graph $G_{W,T'}$ where there is an edge from $i\in W$ to $j\in T'$ if $v_j(s_i^*, \bs_{-i}) \le v_i(s_i^*, \bs_{-i}).$  
 
 Let $G_i$ denote the neighbors of $i\in W$ in the graph, and $B_i = T' \setminus G_i$.  We will use Hall's theorem to show that there is a perfect matching in the graph.  Specifically, for any  $S\subseteq W$, we will show that $|\Gamma(S) | \ge |S|$
where  $\Gamma(S)$ are the neighbors of $S$.

Without loss of generality
  assume that the agents of $W$ are indexed in order of decreasing
  value at $\bs$ (so their indices are $1, \ldots, r$).
For all $k \le r$, and $j < k$, we have
  $v_j(s_k^*, \bs_{-k})\ge v_k(s_k^*, \bs_{-k})$. This is because
  $v_j(\bs) \ge v_k(\bs)$, and the single-crossing condition ensures
  that agent $k$ cannot be above $j$ at $s_k^*$ and end up at or below
  $j$ at $s_k \geq s^*_k$.

For each $i$, let $Y_i$ be the set of agents in $W\cup T'$ which have greater value than $i$ at $(s_i^*, \bs_{-i})$.
Since $s_i^*$ is critical for $i$,
\begin{equation}
\label{Rinc}
rank (Y_i ) < rank (Y_i \cup \{i\}).
\end{equation}
Now, let $S = \{i_1, \ldots, i_k\}$ (ordered in decreasing order of value at $\bs$), let 
$B= \cap_{i\in S} B_i$, and let
$S_j$ be the first $j-1$ elements of $S$.
Then for each element $i_j$ of $S$, by submodularity of the rank function, 
the fact that $S_j \cup B \subseteq Y_{i_j}$ and (\ref{Rinc}), we have
$$rank(S_j \cup B) < rank (S_j \cup B \cup \{i_j\}) \quad\text{ and thus }\quad
rank (S_j \cup B \cup \{i_j\} )= rank (S_j \cup B) + 1.$$
Induction on $j$ then implies that
$rank(S\cup B) = |S| + |B|.$
Since this is at most $r$,  $|B| \le r - |S|$, and so  $|\Gamma(S)| = |T' | - |B|= r - |B| \ge |S|$.

To handle the case that $T'$ is not full rank, simply augment it to a full basis using elements of $W$, and show that there is a matching between $W\setminus T'$ and $T'$ by following the same argument.
\end{proofof}

\section{Conclusions and discussion}
\label{sec:conclusions}
In the previous sections we showed that \GVCGL\ provides constant
factor approximations to the optimal ex post IC revenue in
interdependent settings under certain assumptions. There are several
directions for further work which we now discuss.

\begin{itemize}
\item {\em Prior independence.}
In independent values settings it has been shown that under certain
assumptions (e.g., regularity of the value distributions) it is
possible to design mechanisms that do not require knowing the value
distributions and yet obtain an approximation to the optimal expected
revenue. Such mechanisms are called prior-independent. One way of
designing a prior-independent mechanism, called the single-sample
approach, is to use a reserve price based mechanism (such as GVCG-L)
and replace the reserve price by an independent sample from the
agent's value distribution. The mechanisms that we design require
knowing the conditional distributions of agents' values in order to
determine an appropriate reserve price. Is it possible to design
prior-independent mechanisms in interdependent settings?

It is easy to see that assuming regularity of the conditional
distribution a single-sample approach works: we can replace the
optimal reserve price for an agent in GVCG-L by a random independent
draw from the agent's value distribution conditioned on others'
reported signals, and still obtain a constant factor approximation to
the optimal expected revenue. However, this approach is
unsatisfying. In independent value settings, if there are several
agents with identically distributed values, we can remove one of these
agents at random from the auction and use his value as the random
reserve. In interdependent value settings, we cannot truely ``remove''
an agent from the auction because her signal affects others'
values\footnote{\citet{RT13} show that under
  the so-called Lopomo assumptions, which includes strong MHR and
  symmetry across agents, one can get a constant factor approximation
  by ``reusing'' the signal of an agent to determine a random reserve
  without dropping the agent from the auction.}.

Another approach to prior-independent mechanism design is to
artificially limit supply \citep{DHKN11, RTY12}. In the interdependent values
context, we might ask, for instance: when agents are symmetric and
conditional value distributions satisfy regularity, does the VCG
mechanism with $k/2$ units for sale approximate the expected revenue
of an optimal $k$-unit auction for $k\ge 2$?

\item {\em Relaxing the assumptions.}  In the absence of the
  single-crossing assumption, the VCG mechanism and its variants are
  no longer ex post incentive compatible. While it appears to be
  challenging to characterize the optimal mechanism in this context,
  approximations may be tractable. Likewise, beyond matroid
  feasibility constraints, VCG and its variants are no longer
  necessarily ex post IC.  Is it possible to approximate the optimal
  mechanism in this case?


\item {\em Ex post IC versus Bayesian IC.}
In independent value settings with single-parameter agents, Myerson
shows that ex post IC mechanisms are as strong as BIC mechanisms in
terms of extracting revenue. In interdependent settings, Cremer and
McLean show that BIC mechanisms can extract the entire social
welfare under mild assumptions, which can in general be much larger
than the optimal ex post IC revenue. However, Cremer and McLean's
mechanism violates ex post individual rationality. Can BIC mechanisms
obtain more revenue than ex post IC mechanisms under an ex post
IR constraint? Can we approximate this revenue?

\end{itemize}


\bibliographystyle{apalike}

\appendix

\section{Simple vs.\@ Optimal Results via Lookahead Auctions}
\label{app0}
In this appendix we give quick derivations for several results in the literature
showing that mechanisms simple in form approximates the optimal revenue.

Recall that a \emph{monopoly reserve price} for a bidder (or rather, her value
distribution) is the optimal take-it-or-leave-it price that an auctioneer should
set when selling the item to the bidder alone.  For bidder~$i$, we denote her
monopoly reserve price as $\monres_i$.  

\begin{theorem}
\label{thm:lookahead-vcg}
In general single-dimensional auction settings, when bidders' private valuations are drawn independently from regular
distributions, the VCG-L auction with conditional monopoly reserve
prices is identical to the VCG-L auction with (unconditional) monopoly reserve prices.  
\end{theorem}

\begin{proof}
In the lookahead auction, once a bidder~$i$ is in~$\winset$, the auctioneer runs
an optimal auction for the bidder conditioning on the other bidders' values and
the fact $i \in \winset$ given these values.  In the independent value settings,
this simply amounts to setting an optimal take-it-or-leave-it price for the
bidder conditioning on $\vali \geq \vcgprice_i$.  We show that this optimal
price is equal to $\max \{\vcgprice_i, \monres_i$.  With this, a quick
comparison would confirm that the lookahead auction is identical with VCG-L with
monopoly reserves.


The distribution of $v_i$, conditioning on
$v_i \geq \vcgprice_i$ is just $\dist_i$ truncated at $\vcgprice$, i.e.,
$\conddist_i(v) = \tfrac{\dist_i(v) - \dist(\vcgprice_i)}{1 -
\dist_i(\vcgprice_i)}$, for any $v \geq \vcgprice_i$, and $\conddist_i(v) = 0$
for any $v < \vcgprice_i$.  Setting any price below $\vcgprice_i$ would be obviously suboptimal, and the revenue of setting a price of $\payment \geq
\vcgprice_i$ is 
\begin{align*}
\payment (1 - \conddist_i(\payment)) = \payment \cdot \left( 1 -
\frac{\dist_i(\payment) - \dist_i(\vcgprice_i)}{1 -
\dist_i(\vcgprice_i)} \right) = \payment (1 - \dist_i(\payment)) \cdot
\frac{1}{1 - \dist_i(\vcgprice_i)}.
\end{align*}
In words, for prices above $\vcgprice_i$, the conditional expected revenue is simply scaled up by a constant factor of $\tfrac{1}{1 -
\dist_i(\vcgprice_i)}$ from the unconditioned distribution. Therefore, if
$\vcgprice_i < \monres_i$, the optimal price to set for the conditional distribution is
still $\monres_i$; and if $\vcgprice_i \geq \monres_i$, by regularity, the expected revenue monotonically decreases as the
price rises above $\vcgprice_i$; therefore the optimal price to set is
$\vcgprice_i$ itself.
\end{proof}

Given \autoref{thm:lookahead-matroid-approx}, this immediately implies the following result by \citet{DRY10}.

\begin{corollary}[\citealp{DRY10}]
\label{cor:vcg-el-matroid}
In matroid settings where bidders' valuations are drawn independently from
regular distributions, the VCG-L auction with monopoly reserve prices gives at least half of the optimal revenue.
\end{corollary}

As another illustration of the usefulness between the lookahead auction and the
reserve-based VCG auctions, we use the analysis of the lookahead auction to give
a short re-derivation for a result by \citet{ADMW13}, a major building block in that work.

Let $\rev_i(\payment)$ denote the expected revenue from bidder~$i$ by setting a price of~$\payment$.

\begin{theorem}[Theorem 3.1 in \citealp{ADMW13}]
\label{thm:vcgl-single}
For each bidder~$i$, let $\reserve_i$ be a price drawn randomly from a certain distribution, such that in the
single-bidder setting, $\Ex[\reserve_i]{\rev_i(\reserve_i)} \geq \alpha \rev_i(\monres_i)$.  Then in all downward closed
settings where bidders' valuations are drawn indepdendently from regular distributions, VCG-L with random
reserve prices $(\reserve_1, \ldots, \reserve_n)$ gets at least $\alpha$-fraction of the revenue of VCG-L with monopoly
reserve prices (VCG-LM). 
\end{theorem}

\begin{proof}
Fixing a tentative winner~$i \in \winset$, and condition on the VCG price $\vcgprice_i$.  For a reserve price $\reserve_i$, let $\ratio$ denote the ratio $\rev_i(\reserve_i) /
\rev_i(\monres_i)$, then we have $\Ex{\ratio} \geq \alpha$.  Let $\ratiocond$ denote the ratio between the expected
VCG-L revenue from using reserve price~$\reserve_i$ and that from VCG-LM, conditioning on $v_i \geq \vcgprice_i$.  
It suffices to show $\Ex{\ratiocond} \geq \alpha$.  We will show $\ratiocond \geq \ratio$ pointwise.   Since, as we have
shown in the proof of \autoref{thm:lookahead-spa}, the revenue from all prices above $\vcgprice_i$ is scaled up by a
factor of $1/ (1 - \dist_i(\vcgprice_i))$ when conditioning on $v_i > \vcgprice_i$,  for $\reserve_i \geq \vcgprice_i$ we have
\begin{align*}
\ratiocond = \frac{\rev_i(\reserve_i) / (1 - \dist_i(\vcgprice_i)}{\rev_i(\max \{\vcgprice_i, \monres_i\}) / (1 - \dist_i(\vcgprice_i)} \geq \frac{\rev_i(\reserve_i)}{\rev_i(\monres_i)} = \ratio.
\end{align*}
  For $\reserve_i < \vcgprice_i$, if $\vcgprice_i \geq
\monres_i$, then both auctions will use $\vcgprice_i$, and $\ratiocond = 1 \geq \ratio$; if $\vcgprice_i < \monres_i$,
we have 
\begin{align*}
\ratiocond = \frac{\rev_i(\vcgprice_i) / (1 - \dist_i(\vcgprice_i))}{\rev_i(\monres_i) / (1 - \dist_i(\vcgprice_i))}
= \frac{\rev_i(\vcgprice_i)}{\rev_i(\monres_i)} \geq \frac{\rev_i(\reserve_i)}{\rev_i(\monres_i)} = \ratio.
\end{align*}
The last inequality comes from regularity, i.e., $\reserve_i < \vcgprice_i \leq \monres_i$ implies
$\rev_i(\reserve_i) \leq \rev_i(\vcgprice_i)$.  This completes the proof.
\end{proof}

\section{VCG-L vs.\@ VCG-E}
\label{app1}
In this appendix we compare the revenue of the VCG-E and VCG-L auctions with the
same set of reserves in for independent settings with regular distributions.
Although we end the appendix with \autoref{cor:vcg-e-matroid}, a known result
from \citet{HL09}, the general comparison we make here concerns VCG-L and VCG-E
auctions with reserves that are not necessarily the monopoly reserves; in fact,
even for the monopoly reserve prices, the revenu dominance of VCG-E that we show
here is not known in the literature.  Given the practicality and popularity of
such reserve-based auctions, we believe this comparison to be of independent
interest.


Recall that the VCG-E is the auction that first removes all bidders bidding
below their reserve prices and then run VCG-L on the remaining bidders.  It is
not too hard to see that adding the first stage improves the social welfare, but
the the revenue analysis is nontrivial.  We need to analyze two counteractive
aspects.  On the one hand, VCG-E removes bidders that are
not going to afford their take-it-or-leave-it price later in the auction, and
helps free up resources to be sold to the other bidders; on the other hand, the
existence of more bidders always creates more competition, and makes the VCG
payment for other bidders higher and therefore helps create revenue.
\autoref{thm:eager-vs-lazy} shows that, when the reserve prices are monopoly reserves, the first force is the dominant one.

The analysis is not trivial, and we first present the case of single-item
auctions as a warm-up.  Throughout this appendix we denote by $\monres_i$ the monopoly reserve price for a bidder~$i$, i.e., the optimal posted price one would set when selling a single
item to this bidder.

\begin{theorem}
\label{thm:eager-beats-lazy}
In a single item auction where bidders' valuations are drawn independently from
regular distributions, the second price
auction with eager monopoly reserve prices (VCG-EM) generates weakly more revenue than the second price auction with
lazy monopoly reserve prices (VCG-LM).
\end{theorem}

\begin{proof}
The key idea of the proof is to consider the expected revenue \emph{from each bidder}, conditioning on the other bidders'
valuations.  We calculate the conditional expected revenues from the two auctions and compare them.  The computation of
the conditional revenue does not assume conditions on the reserve price or the valuation distribution; only in the last
step do we use the property of monopoly reserves and the valuation
distribution's regularity.  

Fix a bidder~$i$, we condition on all
other bidders' valuations.  Denote by $\highlazy$ the highest valuation among these other bidders, i.e., $\highlazy =
\max_{j} \{v_j \: | \: j \neq i\}$.  Denote by $\higheager$ the highest valuation among bidders except~$i$ who bid above
their reserve prices, i.e., $\higheager = \max_j \{v_j \: | \: j \neq i, v_j \geq \reserve_j\}$.  Clearly, $\highlazy
\geq \higheager$.

When $\highlazy \leq \reserve_i$, the revenue extracted from bidder~$i$ is exactly the same in VCG-E and VCG-L, because
in both auctions, bidder~$i$ will make a payment if and only if his bid is at least~$\reserve_i$, in which case he
pays~$\reserve_i$.  Therefore we need only consider the case $\highlazy > \reserve_i$.

In VCG-L, bidder~$i$ makes a payment of~$\highlazy$ if and only if his valuation is above~$\highlazy$ (since $\highlazy
> \reserve_i$). The expected
revenue from him is then $\highlazy (1 - \dist_i(\highlazy))$.  In VCG-E, bidder~$i$ makes a payment if and only if
his valuation is above both $\reserve_i$ and $\higheager$, and the payment is $\paymenteager_i = \max \{\reserve_i,
\higheager\}$.  The expected revenue from~$i$ is therefore $\paymenteager_i (1 - \dist_i(\paymenteager_i))$.  Note that
$\highlazy \geq \max \{\reserve_i, \higheager\} = \paymenteager_i$.

But the revenues $\highlazy(1 - \dist_i(\highlazy))$ and $\paymenteager_i (1 - \dist_i(\paymenteager_i))$ are simply the expected
revenue we get by setting a price of $\highlazy$ or $\paymenteager_i$ to
bidder~$i$, respectively.  By regularity of the valuation distribution, the expected revenue monotonically decreases as we raise the price above the monopoly reserve.  
Since $\highlazy \geq \paymenteager_i \geq \reserve_i = \monres_i$, we have
\begin{align*}
\highlazy(1 - \dist_i(\highlazy)) \leq \paymenteager_i (1 - \dist_i(\paymenteager_i)).
\end{align*}

The above is a conditional analysis, but we see that the expected revenue in VCG-LM from each bidder~$i$ is no more than
the expected revenue in VCG-EM from the same bidder, no matter what the other bidders bid.  Our theorem immediately
follows.  
\end{proof}

\begin{theorem}
\label{thm:eager-vs-lazy}
For independent private value settings with matroid feasibility constraints and
regular valuation distributions, the VCG-E auction has weakly better social
welfare than VCG-L, and when the reserves are at least the monopoly reserves,
the VCG-E auction also obtains weakly more revenue than VCG-L.
\end{theorem}

\begin{proof}
As in the proof of \autoref{thm:eager-beats-lazy}, we focus on the expected revenue from a fixed bidder~$i$, conditioning on all
other bidders' valuations.  Let $\paymentvcgl_i$ and $\paymentvcge_i$ denote the VCG payment for bidder~$i$ in VCG-L
and VCG-E, respectively.  (Note that the actual payment to be made by bidder~$i$ in these auctions is the higher of the
reserve~$\reserve_i$ and the VCG threshold; but we will focus mostly on the VCG thresholds in this proof.)  The key step is to show $\paymentvcgl_i \geq \paymentvcge_i$.  Recall that in VCG-E, bidders with
valuations below their reserve prices are excluded from both the allocation and payment calculation.  Effectively, such
bidders' bids are replaced by~$0$ in the calculation of VCG payments.  This constitutes the only difference in computing
$\paymentvcgl_i$ and $\paymentvcge_i$.  It therefore suffices to show that bidder~$i$'s VCG payment weakly decreases as
other bidders' bids are zeroed out.

We show this by expressing $\paymentvcgl_i$ and $\paymentvcge_i$ in terms of a submodular function.  Define a set function $f:
2^{[n]} \to \mathbb R$ by 
\begin{align*}
f(S) = \max_{T \subseteq S, T \in \feasset} \sum_{j \in T} v'_j,
\end{align*} 
where $v'_j = v_j$ for $j \neq i$ and $v'_i = \sum_{j \neq i} v_j + 1$.  Then $\paymentvcgl_i =f([n]
\setminus \{i\}) -  (f([n] -  v'_i)$.  This is because $f([n] \setminus \{i\})$ represents the social welfare of the other
bidders if $i$ were not present, and $f([n]) - v'_i$ represents the social welfare of the other bidders if $i$ were to
be a winner ($v'_i$ is set high enough such that the optimal set has to include~$i$).  The difference between the two is
then the externality that $i$ imposes on the other bidders by his winning, and therefore is equal to the VCG threshold
$\paymentvcgl_i$.  Similarly, let $U$ denote the set of bidders whose valuations are above their reserve prices, i.e.,
$U = \{j \: | \: v_j \geq \reserve_j\}$, then $\paymentvcge_i = f(U) - (f(U \cup \{i\}) - v'_i)$.

$f$ is defined to be the result of a linear maximization over a matroid, and is well known to be submodular.  Therefore,
as $U \subseteq [n]$, $f([n]) - f([n] \setminus \{i\}) \leq f(U \cup \{i\})
- f(U)$.  This gives $\paymentvcgl_i \geq \paymentvcge_i$.

With this, the proof to \autoref{thm:eager-beats-lazy} easily generalizes by replacing $\highlazy$ and $\higheager$ with
$\paymentvcgl_i$ and $\paymentvcge_i$, respectively, and we omit the rest of the proof.
\end{proof}

\begin{corollary}[\citealp{HR09}]
 \label{cor:vcg-e-matroid}
In matroid settings where bidders' valuations are drawn independently from
regular distributions, the VCG-E auction with monopoly reserve prices gives at least half of the optimal revenue.
\end{corollary}

\begin{remark}
\label{remark:non-matroid}
For more general environments (even downward closed), it need not be true that $\paymentvcgl_i$ is at least
$\paymentvcge_i$.  For example, consider $\feasset = \{\emptyset, \{1\}, \{2\}, \{1, 2\}, \{3\}\}$.  In this
feasibility system, bidder~$1$'s VCG payment weakly increases when $v_2$ decreases.  So $\paymentvcgl_1$ can
be lower than $\paymentvcge_1$.
\end{remark}

\end{document}